%% file: chan.tex
\documentclass[a4paper]{article}

\usepackage{amsthm}
\usepackage{amsmath}
\usepackage{amssymb}
\usepackage{amsfonts}
\usepackage{mathrsfs}
\usepackage{verbatim}
\usepackage{todonotes}

\newtheorem{theorem}{Theorem}
\newtheorem{lemma}[theorem]{Lemma}

\newtheorem{definition}[theorem]{Definition}

\begin{document}
\title{Improved Hardness of Approximating Chromatic Number}
\author{Sangxia Huang \\ KTH Royal Institute of Technology \\ Stockholm, Sweden \\ \texttt{sangxia@csc.kth.se}}
\maketitle

\begin{abstract}
  We prove that for sufficiently large $K$, it is NP-hard to color
  $K$-colorable graphs with less than $2^{K^{1/3}}$ colors. This improves
  the previous result of $K$ versus $K^{O(\log K)}$ in Khot \cite{khotcoloring}.
\end{abstract}

\input{intro}

\input{prelim}

\input{coloring}

\section*{Acknowledgements}
The author is grateful to Siu On Chan for pointing out Theorem E.1 
in his paper, leading to a gap of $K$ vs. $2^{K^{1/3}}$, improving the original
gap of $K$ vs. $2^{K^{1/5}}$.

\bibliographystyle{abbrv}
\bibliography{bib}

\end{document}

%% file: intro.tex
\section{Introduction}
In this paper we improve the inapproximability result for approximating the chromatic number of
graphs. More specifically, given a $K$-colorable graph, we would like to color it with as few
color as possible.
This problem is closely related to a number of other problems such as approximation of
independent sets and PCPs with low amortized free bit complexity. On the algorithmic side,
Karger, Motwani and Sudan \cite{kargermotwanisudan} give an algorithm
to color a $K$-colorable graph with $\tilde{O}(n^{1-3/(K+1)})$ colors,
whereas Blum and Karger \cite{blumkarger} gave an algorithm to color a 3-colorable
graph with $\tilde{O}(n^{3/14})$ colors.
There have been many works on the hardness side as well.
It is known that coloring 3-colorable graph with 5 colors is NP-hard, and for general
$K$-colorable graph it is NP-hard to color with $K+2\lceil \frac{K}{3} \rceil$ colors
\cite{hard3, hard5}.
For all sufficiently large $K$, the best known gap is by Khot \cite{khotcoloring}
which proved that it is NP-hard to color a $K$-colorable graph with $K^{O(\log K)}$ colors.
This is equivalent to saying that it is NP-hard to distinguish a $K$-colorable graph from
a graph that cannot be colored with fewer than $K^{O(\log K)}$ colors.
Assuming a variant of Khot's 2-to-1 Conjecture, Dinur, Mossel and Regev \cite{dmrcoloring} proved
that it is NP-hard to $K'$-color a $K$-colorable graph for any $3 \le K < K'$.

Khot's hardness result \cite{khotcoloring} can be derived either using PCPs from H{\aa}stad and Khot \cite{hastadkhot}
or Samorodnitsky and Trevisan \cite{stfirst}. The PCPs is then combined with a randomization technique to show
the completeness of the reduction. We can view the results in \cite{hastadkhot} and \cite{stfirst} as showing approximation
resistance for a family of Boolean predicates that has very few accepting inputs
--- it is NP-hard to approximate Max CSP with those predicates better than just doing random assignments.
For each integer $k>0$, the approximation resistance predicates we get from \cite{hastadkhot} and \cite{stfirst}
has $k$ variables (and thus $2^k$ possible assignments) but only has $2^{O(\sqrt{k})}$ accepting assignments.
The predicate \cite{hastadkhot} is approximation resistance even on satisfiable instances --- or perfect completeness in PCP language
--- while the predicate from \cite{stfirst} is not.
It is noted in \cite{khotcoloring} that having perfect completeness makes the reduction for coloring easier.

In a recent breakthrough, Chan \cite{chan} proved approximation resistance for a family of predicates on $k$
variables but only has $k+1$ accepting assignments whenever $k$ is of the form $k=2^r-1$. 
Previously, approximation resistance of those predicates are only known assuming the Unique Games Conjecture
\cite{hadamardugc}.
Hast \cite{hast} proved that predicates on $k$ variables having at most $2\lfloor k/2 \rfloor+1$ ($=k$ in the
current setting) accepting inputs are not approximation resistant.
This hardness also matches up to constant factor with
the algorithm for Max CSP by Charikar, Makarychev and Makarychev \cite{cmm09}.

It is an interesting question if this improved result can be used as in \cite{khotcoloring} to get better
approximation gap for chromatic numbers. In the following theorem, we show that this is indeed the case.
\begin{theorem}
  For all sufficiently large $K$, it is NP-hard to color a $K$-colorable graph with $2^{K^{1/3}}$ colors. 
\end{theorem}

In \cite{chan}, Chan also showed that for any $K \ge 3$, there is $\nu=o(1)$ such that given a graph with an induced
$K$-colorable subgraph of fractional size $1-\nu$, it is NP-hard to find an independent set of fractional size $1/2^{K/2}+\nu$.
We refer to \cite{dkps10, ks12, chan} for additional discussions on Almost-Coloring.

%% file: prelim.tex
\section{Label Cover and PCP}
In this section we review basics of Label Cover and PCPs and describe Chan's improved PCP construction.

Let $(U,V,E,L,R,\Pi)$ be an instance of Label Cover, where $R=dL$ for some constant $d$, 
the tuple $(U,V,E)$ is a weighted
bipartite graph, vertices in $U$ are assigned labels from $[L]$, and vertices in $V$ 
are assigned labels from $[R]$.
Each edge $e=(u,v)$ is associated with $d$-to-1 mapping $\pi_e: [R] \to [L]$.
Given an assignment $A: U \to [L], V \to [R]$, constraint on $e$ is satisfied
if $\pi_e(A(v))=A(u)$. 
The following theorem combines the celebrated PCP theorem \cite{arorapcp1, arorapcp2} with Raz's parallel repetition theorem \cite{raz}
and shows hardness of Label Cover.
\begin{theorem}
  \label{thm:labelcover}
  For any constant $0<\sigma<1$,
  there are $d,R \le \mathrm{poly}(1/\sigma)$ such that the problem of deciding a 3-SAT instance with $n$ variables
  can be Karp-reduced in $poly(n)$ time to the problem of $(1,\sigma)$-deciding a Label Cover instance
  of size $n^{1+o(1)}$. Furthermore, $L$ is a bi-regular bipartite graph with left- and right-degrees
  $\mathrm{poly}(1/\sigma)$.
\end{theorem}

As is the case with many MaxCSP inapproximability results, the above Label Cover will be the starting point
of our reduction.
Typically, the reduction translates labelings
for $u \in U$ and $v \in V$ to $2^{|L|}$ and $2^{|R|}$ Boolean variables,
respectively. These variables are viewed as functions
$f^u: \{-1,1\}^{|L|} \to \{-1,1\}$ and $g^v: \{-1,1\}^{|R|} \to \{-1,1\}$.
We require that these functions are folded, that is, 
for any $x \in \{-1,1\}^{|L|}$, $y \in \left\{ -1,1 \right\}^{|R|}$,
$f^u(-x)=-f^u(x)$ and $g^v(-y)=-g^v(y)$. This corresponds to having
negated literals in the CSP instances.
In a correct proof for a satisfiable Label-Cover instance, 
we expect the functions to be long codes for the corresponding
labelings of $u$ and $v$, that is, setting $f^u({x}) = x_{\sigma_U(u)}$,
and $g^v({y}) = y_{\sigma_V(v)}$.

For an edge $(u,v)$ in the Label-Cover, we sample {\em queries}
\[({x}^{(1)}, \cdots, {x}^{(m)}, {y}^{(m+1)}, \cdots, {y}^{(k)})\]
according to some carefully chosen {\em test distribution} $\mathcal{T}$.
The distribution $\mathcal{T}$ has the property that for any $l \in L$ and $r \in R$
such that $\pi_{(u,v)}(r)=l$, the predicate $P$ accepts
\[
(f_u({x}^{(1)}_{l}), \cdots, f_u({x}^{(m)}_{l}), g_v({y}^{(m+1)}_{r}), \cdots, g_v({y}^{(m+1)}_{r}))
\]
with probability 1 (or $1-\varepsilon$ for some small constant $\varepsilon$ if we
are considering non-perfect completeness).
We can think of this as a 2-player game, functions $f_u$ and $g_v$ are strategies of each player,
and the goal is to convince the verifier who uses predicate $P$ to decide acceptance.

In \cite{chan}, Chan developed a new way of constructing efficient PCPs and proved that the following
Hadamard predicate $H_K: \{-1,1\}^K \to \{0,1\}$ is approximation resistant
for $K=2^r-1$.
The predicate $H_K$ is on variables $\{x_S\}_{\emptyset \ne S \subseteq [k]}$, defined as
\[
H_K(x) = \left\{
\begin{array}{l l}
  1 & \forall S \subseteq [k], |S|>1, x_S = \prod_{i \in S} x_{\{i\}}\\
  0 & \mathrm{otherwise.}
\end{array}
\right.
\]
This predicate has $K+1$ accepting assignments.
Samorodnitsky and Trevisan \cite{hadamardugc} showed that $H_K$ is 
approximation resistance assuming the Unique Games Conjecture.
Using his new technique, Chan proved that this is true assuming $P \ne NP$.

The main idea in Chan's reduction is to consider a direct sum of $K$ independent $K$-player games.
In the $i$-th game, an edge $e_i$ is sampled as in the above 2-player game. Player $i$ gets
a uniform random string from $\{-1,1\}^{|L|}$ and all other players get samples from $\{-1,1\}^{|R|}$
as described below. Thus for each tuple of $K-1$ vertices from $V$ and 1 vertex from $U$,
the players will have a strategy which is a Boolean function
taking as input $(K-1)$ strings of length $R$ and 1 string of length $L$. 
In a correct proof, the strategies are expected to be products of long codes encoding the labeling of the vertices.

We now formally define the PCP and how queries are sampled.
\begin{definition}
  \label{def:pcpreduction}
  Let $(U,V,E,L,R,\Pi)$ be a label cover instance. 
  Define $\mathcal{V}_i = V^{i-1} \times U \times V^{K-i}$ for $i \in [K]$.
  For each $\mathbf{v} \in \mathcal{V}_i$, the proof contains function 
  $\mathbf{f}_{\mathbf{v}}: \left(\{-1,1\}^{R}\right)^{i-1} \times \{-1,1\}^{L} \times 
  \left(\{-1,1\}^{R}\right)^{k-i} \to \{-1,1\}$.
  The verifier check the proof as follows:
  \begin{enumerate}
    \item Sample independently $K$ random edges $e_1=(u_1,v_1), \cdots, e_K=(u_K,v_K) \in E$.
    \item Let $\mathbf{v}_i = (v_1,\cdots,v_{i-1},u_i,v_{i+1},\cdots,v_{K})$. The verifiers queries
      $\{\mathbf{f}_{\mathbf{v}_i}(\mathbf{q}_i)\}_{i=1}^{K}$, where $\mathbf{q}_i$ is sampled as below.

      For $i \in [K]$, first sample uniformly $\mathbf{q}_{i,i} \in \{-1,1\}^{L}$. 
      Let $\pi$ be the projection of edge $e_i$.
      For each $r \in R$, pick a uniformly random point
      $(x_1, \cdots, x_K)$ from the subspace of $H_K$
      where $x_i=\mathbf{q}_{i,i,\pi(r)}$, and set $\mathbf{q}_{i,j,r}=x_j$ for $j \ne i$.
      Finally for each bit in $\mathbf{q}_i$, $i = 1, \cdots, K$, with probability $\eta$ we resample the bit from the uniform
      distribution on $\{-1,1\}$.
    \item 
      Let $\mathbf{f}(\mathbf{q}) = 
        (\mathbf{f}_{\mathbf{v}_1}(\mathbf{q}_1), \cdots, \mathbf{f}_{\mathbf{v}_K}(\mathbf{q}_K))$.
        Accept if $H_K(\mathbf{f}(\mathbf{q}))=1$.
  \end{enumerate}
\end{definition}
In a correct proof, the function $\mathbf{f}_{\mathbf{v}}$ is the product of long codes encoding the labeling
of each vertex in $\mathbf{v}$.

Theorem E.1 along with Theorem A.1, 6.9 and C.2 of \cite{chan} shows completeness and soundness of the above reduction
and we formulate as the following theorem.
\begin{theorem}
  \label{thm:chan}
  Fix some small $\eta, \delta>0$.
  Let the soundness of Label Cover $\sigma$ be chosen such that
  $\delta = \textrm{poly}(K/\eta) \cdot \sigma^{\Omega(1)}$.
  Given a Label Cover instance $LC_{L,dL}$, we have the following:
  \begin{enumerate}
    \item If $LC_{L,dL}$ has value 1, the verifier accepts a correct proof with probability at least $1-K^2 \eta$.
    \item If the verifier accepts with probability greater than $(K+1)/2^K+2\delta$, then $LC_{L,dL}$ has value at least $\sigma$.
  \end{enumerate}
\end{theorem}

%% file: coloring.tex
\section{Hardness of Approximating Chromatic Number}
In this section, we prove that for sufficiently large $K$, it is NP-hard to color
a $K$-colorable graph with less than $2^{K^{1/3}}$ colors.
For convenience of notation, we in fact prove a gap of $K^3$ versus $2^K$.

The overall idea is similar to that in Khot \cite{khotcoloring}.
We start by describing the FGLSS graph \cite{fglss} of the PCP in Definition \ref{def:pcpreduction}.
The vertices in the FGLSS graph are function
queries and corresponding accepting
configurations, denoted as $(\mathbf{f}_{\mathbf{v}},\mathbf{q}, \mathbf{x})$. 
The weight of the vertex is the probability that query $(\mathbf{f}_{\mathbf{v}},\mathbf{q})$ is picked.
The total weight of the graph is therefore $K+1$.
Two vertices are connected if they are clearly inconsistent
(returning different answers for the same query to the same function).
An independent set in the graph corresponds to a strategy / set of functions,
and its weight is the acceptance probability of such strategy.
For a PCP, the maximum weight independent set in its FGLSS graph 
is related to its acceptance probability. Note that if the maximum 
weight independent set has weight $w$, then we need at least $(K+1)/w$ colors to color
the whole graph since vertices having the same color must form an independent set.

To use the FGLSS graph for coloring results, we also need to show that if a PCP has acceptance probability
$1-\varepsilon$, we can color the FGLSS graph with a small number of colors.
As in \cite{khotcoloring}, we modify Definition \ref{def:pcpreduction} so that the functions in the proof become
$\mathbf{f}_{\mathbf{v}}: \left(\{-1,1\}^{R \cdot 2^{t}}\right)^{i-1} \times \{-1,1\}^{L \cdot 2^{t}}
\times \left(\{-1,1\}^{R \cdot 2^{t}}\right)^{k-i} \to \{-1,1\}$.
Alternatively, we can think of this as modifying Label Cover by appending a $t$-bit binary string to all the labels
and defining the new projection in the Label Cover instance
as $\pi_e'(r \circ \alpha)=\pi_e(r) \circ \alpha$ for $r \in R$ and $\alpha \in \{0,1\}^{t}$,
where ``$\circ$'' denotes string concatenation.
The value of this new Label Cover instance is exactly the same as the original setting.

Consider the FGLSS graph of the PCP based on the above modified Label Cover problem.
Soundness is straightforward. If the new proof makes the verifier accept with probability at least
$(K+1)/2^K+2\delta$, then the value of the new Label Cover is at least $\sigma$ and thus so is the original instance.

Let $\alpha \in \{0,1\}^t$ be some global parameter. If a Label Cover instance has value 1,
then there are at least $2^t$ different ways of assigning labels satisfying all the edges.
In a correct proof, functions would be product of long codes encoding labelings extended by
$\alpha$, denoted as $\mathbf{f}^{(\alpha)}$.
Next we prove that independent sets corresponding to different $\alpha \in \{0,1\}^t$ covers almost
all of the FGLSS graph of the modified PCP.
We remove the small fraction of vertices that are not covered and this gives the completeness of the reduction.

Formally, we follow Khot's notation and introduce 
the following definition characterizing whether we can cover certain vertex with
independent sets.

\begin{definition}
  \label{def:goodquery}
  A set of queries $\mathbf{q}$ is good if for any labelings to the vertices
  and any accepting assignment $\mathbf{z}$,
  there exists a global extension $\alpha$, such that if the functions $\mathbf{f}^{(\alpha)}$
  are products of long codes encoding the labelings extended by $\alpha$, then 
  $\mathbf{f}^{(\alpha)}(\mathbf{q})=\mathbf{z}$.
\end{definition}

Let $\delta = 2^{-\Omega(K)}$ be the soundness parameter of the PCP.
By Theorem \ref{thm:chan} of Chan, we require the soundness of Label Cover
to be $\sigma = (\delta / \textrm{poly}(K/\eta))^{o(1)}$.

\begin{lemma}
  \label{lem:mainlemma}
  Let $t$ be such that $2^t=C \cdot K^{3}$ for some large constant $C$. For large enough $K$, 
  at most a weighted fraction of $\exp(-O(K))$ of the queries is not good.
\end{lemma}
Remove the vertices in the FGLSS graph that correspond to queries that are not good.
The fraction of vertices removed is bounded by $\exp(-O(K))$.
In the soundness case coloring the FGLSS graph still
needs at least $K(1-\exp(-O(K)))/2^{-O(K)}=2^{\Omega(K)}$ colors.
In the completeness case, each $\alpha \in \{0,1\}^t$ is associated with an independent set
consisting of vertices the form $(\mathbf{f}_{\mathbf{v}},\mathbf{q},\mathbf{z})$, where 
$\mathbf{z}=\mathbf{f}^{(\alpha)}(\mathbf{q})$ and $\mathbf{f}^{(\alpha)}$ is a product of long codes
encoding labelings of $\mathbf{v}$ extended by $\alpha$.
Consider any vertex $(\mathbf{f}_{\mathbf{v}},\mathbf{q},\mathbf{x})$ in the modified FGLSS graph.
By definition $\mathbf{q}$ is good so there exists $\alpha_0 \in \{0,1\}^t$
such that $\mathbf{f}_{\mathbf{v}}(\mathbf{q})=\mathbf{f}^{(\alpha_0)}(\mathbf{q})=\mathbf{x}$,
so it is covered by the independent set associated with $\alpha_0$.
Therefore the modified FGLSS graph can be colored with $2^t=O(K^{3})$ colors.

\begin{proof}[Proof of Lemma \ref{lem:mainlemma}.]
  Let $\eta = O(1/K^2)$. We thus have that the label size $L = \mathrm{poly}(1/\sigma)=\exp(O(K))$.

  Fix some labeling of the label cover instance and some accepting assignment $\mathbf{z}$. 


  Consider $\alpha \in \{0,1\}^t$. 
  Over the queries sampled,
  the probability that $\mathbf{f}^{(\alpha)}(\mathbf{q})=\mathbf{z}$
  is $O(1/K \cdot (1-\eta)^{K^2}) = O(e^{-\eta K^2}/K)=O(1/K)$
  --- the probability of answering $\mathbf{z}$ on $\mathbf{q}$ before adding noise is $1/(K+1)$, 
  there are $K$ functions,
  each being a product of $K$ long codes, therefore 
  the answer $\mathbf{f}^{(\alpha)}(\mathbf{q})$ depends on $K^2$ bits,
  and if none of the $K^2$ bits are corrupted then the answer is correct.
  The contribution from other sources to answer $\mathbf{z}$ is negligible.
  Note that for different extension $\alpha$, the bits that $\mathbf{f}^{(\alpha)}$ reads
  from $\mathbf{q}$ are completely different and therefore independent, so we have
  \[
  \Pr_{\mathbf{q}}\left[\forall \alpha, \mathbf{f}^{(\alpha)}(\mathbf{q}) \ne \mathbf{z}\right] = 
  (1-O(1/K))^{2^t} = \exp(-O(2^t/K)).
  \]

  For query $\mathbf{q}$, let $Q(\mathbf{q})$ be the event that $\mathbf{q}$ is not good
  in the sense of Definition \ref{def:goodquery}: there exists
  some labeling and some accepting assignment $\mathbf{z}$, such that 
  $\mathbf{f}^{(\alpha)}(\mathbf{q}) \ne \mathbf{z}$ for any $\alpha$.
  Taking union bound over all possible labelings and accepting configurations, we get 
  that the weighted fraction of $\mathbf{q}$ that are bad is
  \[
  \Pr_{\mathbf{q}}\left[ Q(\mathbf{q}) \right]  \le L^{O(K)} K \exp(-O(2^t/K)) = \exp(-O(K)).
  \]
\end{proof}